\documentclass[a4paper,11pt]{article}
\usepackage{amsmath,graphicx,amsthm,latexsym,amssymb,natbib}

\newtheorem{theorem}{Theorem}[section]
\newtheorem{lemma}{Lemma}[section]
\theoremstyle{remark}
\newtheorem{remark}[lemma]{Remark}

\title{On data depth in infinite dimensional spaces}
\author{\vspace{0.3in} Anirvan Chakraborty\thanks{Research of A. Chakraborty is partially supported by CSIR SPM Fellowship} ~and Probal Chaudhuri}
\date{}
\begin{document}

\maketitle
\vspace{-0.6in}
\begin{center}
 Theoretical Statistics and Mathematics Unit, \\
 Indian Statistical Institute \\ 
 203, B. T. Road, Kolkata - 700108, INDIA. \\
 emails: anirvan\_r@isical.ac.in, probal@isical.ac.in
\end{center}
\vspace{0.15in}
\begin{abstract}
The concept of data depth leads to a center-outward ordering of multivariate data, and it has been effectively used for developing various data analytic tools.  While different notions of depth were originally developed for finite dimensional data, there have been some recent attempts to develop depth functions for data in infinite dimensional spaces. In this paper, we consider some notions of depth in infinite dimensional spaces and study their properties under various stochastic models. Our analysis shows that some of the depth functions available in the literature have degenerate behaviour for some commonly used probability distributions in infinite dimensional spaces of sequences and functions. As a consequence, they are not very useful for the analysis of data satisfying such infinite dimensional probability models. However, some modified versions of those depth functions as well as an infinite dimensional extension of the spatial depth do not suffer from such degeneracy, and can be conveniently used 
for analyzing infinite dimensional data. 
\vspace{0.1in} 
\end{abstract}

\section{Introduction} 
\label{intro} \indent In finite dimensional spaces, depth functions provide a center-outward ordering of the points in the sample space relative to a given probability distribution, and various depth functions for probability distributions in $\mathbb{R}^{d}$ have been proposed in the literature (see, e.g.,  \cite{LPS99} and \cite{ZS00a} for some extensive review). Several desirable properties of depth functions have been listed in \cite{ZS00a}, and these properties have been utilized in developing several statistical procedures. Depth-weighted L-type location estimators like trimmed means have been considered in \cite{DG92}, \cite{FM01}, \cite{Mosl02} and \cite{Zuo06}. Depth functions have also been used to construct statistical classifiers (see, e.g., \cite{Jorn04}, \cite{GC05a}, \cite{MH06}, \cite{DG12} and \cite{LCAL12}). Another useful application of depths is in constructing depth contours (see, e.g., \cite{DG92} and \cite{Mosl02}), which determine central and outlying regions of a probability 
distribution. These contours and regions are useful in outlier detection.  \\
\indent With the recent advancement of scientific techniques and measurement devices, we increasingly come across data that have dimensions much larger than the sample sizes. Such data cannot be handled using standard multivariate techniques due to their high dimensionalities and low sample sizes. A common approach for handling such data is to embed them into suitable infinite dimensional spaces (e.g., data lying in function spaces). Half-space depth (HD) (see, e.g., \cite{DG92}), projection depth (PD) (see, e.g., \cite{ZS00a}) and spatial depth (SD) (see, e.g., \cite{VZ00} and \cite{Serf02}), which were originally defined for data in finite dimensional spaces, can have natural extensions into infinite dimensional spaces as we shall see in subsequent sections. \\
\indent \cite{FM01} defined a notion of depth, which is called integrated data depth (ID), in function spaces. \cite{FM01} used this depth function to construct trimmed means, and they showed that the empirical ID is a strongly and uniformly consistent estimator of its population counterpart. These authors used ID to categorize extremal and central curves in the data consisting of $100$ curves used to build the NASDAQ $100$ index. Recently, \cite{LPR09, LPR11} introduced two different notions of data depth for functional data, and they called them band depth (BD) and half-region depth (HRD). These authors have used these depth functions for detecting the central and the peripheral sample curves of some real datasets including daily temperature curves for Canadian weather stations and gene expression data for lymphoblastic leukemia. Trimmed means based on BD have been discussed in \cite{LPR06}, and they used it to construct classifiers based on certain distance measures. The distance of an observation from a 
class was defined either as the distance from the trimmed mean of the class or as a trimmed weighted average of the distances from observations in the class. The procedure was implemented to classify the well-known Berkeley growth data (see \cite{RS05}). \cite{LPR09} also proposed a rank based test for two-population problems using BD, and they used the procedure to test the equality of curves obtained by plotting relative diameters along the y-axis against relative heights along the x-axis for two groups of trees as well as the Berkeley growth data. These authors proved that the empirical versions of both of BD and HRD converge uniformly almost surely to their population counterparts. However, it was observed by them that both the depth functions tend to take small values if the sample consists of irregular (non-smooth) curves that cross one another often. To overcome this problem, \cite{LPR09, LPR11} proposed modified versions of these depth functions, called modified band depth (MBD) and modified half-
region depth (MHRD), using the ``proportion of time'' a sample curve spends inside a band or a half-region, respectively.  \\
\indent It was observed in \cite{Liu90} that the maximum value of simplicial depth of a point in $\mathbb{R}^{d}$ with respect to any angularly symmetric absolutely continuous probability distribution is $2^{-d}$. Consequently, the simplicial depth of any point in $\mathbb{R}^{d}$ under such a distribution converges to zero as $d$ grows to infinity. This observation motivated us to critically investigate the behaviour of the above-mentioned depth functions for some standard probability models that are widely used for data in infinite dimensional spaces. It will be shown that infinite dimensional extensions of HD and PD have degenerate behaviour in infinite dimensional spaces. Moreover, both BD and HRD suffer from degeneracy for some standard probability models in function spaces. However, their modified versions as well as ID do not suffer from any such degenerate behaviour for similar probability distributions in function spaces. We also extend the notion of SD into infinite dimensional spaces, and it is 
shown that such an extension leads to a well-behaved and statistically useful depth function for many infinite dimensional probability distributions.

\section{Depths using linear projections} 
\label{sec:1} In this section, we shall consider depth functions that are defined using linear projections of a random element ${\bf X}$. We begin by recalling that in finite dimensional spaces, the definitions of both of HD and PD involve distributions of linear projections of ${\bf X}$. An extension of HD into Banach spaces has been considered in \cite{DGC11}. Consider a Banach space ${\cal X}$, the associated Borel $\sigma$-field, a random element ${\bf X} \in {\cal X}$ and a fixed point ${\bf x} \in {\cal X}$. The HD of ${\bf x}$ with respect to the distribution of ${\bf X}$ is defined as $HD({\bf x}) = \inf\{P({\bf u}({\bf X} - {\bf x}) \geq 0) : {\bf u} \in {\cal X}^{*}\}$, where ${\cal X}^{*}$ denotes the dual space of ${\cal X}$. The PD of ${\bf x}$ with respect to the distribution of ${\bf X}$ is defined as
\begin{eqnarray*}
 PD({\bf x}) = \left[1 + \sup_{{\bf u} \in {\cal X}^{*}} \frac{|{\bf u}({\bf x}) - \theta({\bf u}({\bf X}))|}{\sigma({\bf u}({\bf X}))} \right]^{-1}, 
\end{eqnarray*}
where $\theta(.)$ and $\sigma(.)$ are some measures of location and scatter of the distribution of ${\bf u}({\bf X})$. \\
\indent If ${\cal X}$ is a separable Hilbert space, ${\cal X}$ is isometrically isomorphic to $l_{2}$, the space of all square summable sequences. In that case, ${\cal X} = {\cal X}^{*} = l_{2}$, and ${\bf u}({\bf X})$ and ${\bf u}({\bf x})$ in the definitions of HD and PD given above are same as $\langle{\bf u},{\bf X}\rangle$ and $\langle{\bf u},{\bf x}\rangle$, respectively. Here $\langle.,.\rangle$ denotes the usual inner product in $l_{2}$. We shall first consider the space $l_{2}$ equipped with its usual norm and the associated Borel $\sigma$-field. Consider a random sequence ${\bf X} = (X_{1},X_{2},\ldots) \in l_{2}$ such that $\sum_{k=1}^{\infty} E(X_{k}^{2}) < \infty$, which implies $E({\bf X}) = (E(X_{1}),E(X_{2}),\ldots) \in l_{2}$. Let us set $Y_{1} = X_{1} - E(X_{1})$, and denote by $Y_{k}$ the residual of linear regression of $X_{k}$ on $(X_{1}, X_{2}, \ldots,  X_{k-1})$ for $k \geq 2$. In other words, for $k \geq 2$, $Y_{k} = X_{k} - \beta_{0k} - \sum_{j=1}^{k-1} \beta_{jk}X_{j}$, where $\beta_{0k} + \sum_{j=1}^{k-1} \beta_{jk}X_{j}$ is the linear regression of $X_{k}$ on $(X_{1},X_{2},\ldots,X_{k-1})$. Thus, ${\bf Y} = (Y_{1},Y_{2},\ldots)$ is a sequence of uncorrelated random variables with zero means. Further, since $\tau_{k}^{2} = E(Y_{k}^{2}) \leq E(X_{k}^{2})$ for all $k \geq 1$, we have $\sum_{k=1}^{\infty} \tau_{k}^{2} < \infty$, and hence, ${\bf Y} \in l_{2}$ with probability one. We now state a theorem that establishes a degeneracy result for both of HD and PD under appropriate conditions on the distribution of ${\bf Y}$.
\vspace{-0.05in}
\begin{theorem} \label{thm2.1}
Let $\mu$ denote the probability distribution of ${\bf X}$ in $l_{2}$. Assume that the residual sequence ${\bf Y}$ obtained from ${\bf X}$ is $\alpha$-mixing with the mixing coefficients $\{\alpha_{k}\}$ satisfying $\sum_{k=1}^{\infty} \alpha_{k}^{1-1/2p} < \infty$ for some $p \geq 1$. Further, assume that $\tau_{k} > 0$ for all $k \geq 1$, and $\sup_{k \geq 1} E\{(Y_{k}/\tau_{k})^{2r}\} < \infty$ for some $r > p$. Then, $HD({\bf x}) = PD({\bf x}) = 0$ for all ${\bf x}$ in a subset of $l_{2}$ with $\mu$-measure one. Here $HD({\bf x})$ and $PD({\bf x})$ denote the half-space and the projection depths of ${\bf x}$ with respect to $\mu$, respectively, and in the definition of $PD({\bf x})$, we choose $\theta(.)$ and $\sigma(.)$ to be the mean and the standard deviation, respectively. 
\end{theorem}
\vspace{-0.05in}
\indent It is obvious that for any Gaussian probability measure $\mu$, the assumptions in the preceding theorem hold. Recently, it has been shown in \cite{DGC11} that HD has degenerate behaviour when the probability distribution of ${\bf X} = (X_{1},X_{2},\ldots)$ is such that $X_{1}, X_{2}, \ldots$ are independent random variables satisfying suitable moment conditions. Note that if ${\bf X} = (X_{1},X_{2},\ldots)$ is a sequence of independent random variables with zero means, we have ${\bf Y} = {\bf X}$. In that case, if we choose $p = 1$ and $r = 2$, the moment assumption in the above theorem implies that $\sum_{k=1}^{\infty} E\{(X_{k}/\sigma_{k})^4\}/k^{2} < \infty$, which is the condition assumed in Theorem 3 in \cite{DGC11}. It is worth mentioning here that the above result is actually true whenever $\sum_{k=1}^{\infty} (Y_{k}/\tau_{k})^{2} = \infty$ with probability one (see the proof in Section \ref{sec:5}). This, for instance, holds whenever ${\bf Y}$ is a sequence of independent non-degenerate random variables. The moment and the mixing assumptions on ${\bf Y}$ stated in the theorem are only sufficient to ensure $\sum_{k=1}^{\infty} (Y_{k}/\tau_{k})^{2} = \infty$ with probability one, but by no means they are necessary.  \\
\indent The degeneracy of HD and PD stated in the previous theorem is not restricted to separable Hilbert spaces only. Let us consider the space $C[0,1]$ of continuous functions defined on $[0,1]$ along with its supremum norm and the associated Borel $\sigma$-field. Recall that the dual space of $C[0,1]$ is the space of finite signed Borel measures on $[0,1]$ equipped with its total variation norm. The following result shows the degeneracy of HD and PD for a class of probability measures in $C[0,1]$.
\vspace{-0.05in}
\begin{theorem} \label{thm2.2}
Consider a random element ${\bf X}$ in $C[0,1]$ having a Gaussian distribution with a positive definite covariance kernel, and let $\mu$ denote the distribution of ${\bf X}$. Then, $HD({\bf x}) = PD({\bf x}) = 0$ for all ${\bf x}$ in a subset of $C[0,1]$ with $\mu$-measure one. Here we denote the half-space and the projection depths of ${\bf x}$ with respect to $\mu$ by $HD({\bf x})$ and $PD({\bf x})$, respectively, and we choose $\theta(.)$ as the mean and $\sigma(.)$ as the standard deviation in the definition of $PD({\bf x})$. 
\end{theorem}
\vspace{-0.05in}
\indent The degeneracy of HD stated in Theorems \ref{thm2.1} and \ref{thm2.2} can be interpreted as follows. Let ${\cal X}$ be either $l_{2}$ or $C[0,1]$. Then, for any ${\bf x} \in {\cal X}$, we can choose a hyperplane in ${\cal X}$ through ${\bf x}$ in such a way that the probability content of one of the half-spaces is as small as we want. On the other hand, the degeneracy result about PD in the above theorems implies that one can find an element ${\bf u} \in {\cal X}^{*}$ so that the distance of ${\bf u}({\bf x})$ from the mean of ${\bf u}({\bf X})$ relative to the standard deviation of ${\bf u}({\bf X})$ will be as large as desired. Such degenerate behaviour of HD and PD clearly implies that they are not suitable for center-outward ordering of the points in ${\cal X}$, and these depth functions cannot be used to determine the central and the outlying regions for many probability measures including Gaussian distributions in ${\cal X}$. One reason for such degeneracy is that the dual space ${\cal X}^{*}$ is too large, and its unit ball is not compact (see also \cite{MP12}). It will be appropriate to note here that unlike what we have mentioned about simplicial depth in the Introduction, it is easy to verify that the maximum values of HD and PD for any symmetric probability distribution in ${\cal X}$ such that any linear function has a continuous distribution are $1/2$ (see, e.g., \cite{DGC11}) and $1$, respectively, and these maximum values are achieved at the centre of symmetry of the probability distribution. In other words, although HD and PD have degenerate behaviour in ${\cal X}$, the half-space median and the  projection median remain well-defined for symmetric distributions in ${\cal X}$. \\
\indent Let us now consider a simple classification problem, which involves class distributions in ${\cal X}$ (${\cal X} = l_{2}$ or $C[0,1]$ as in the preceding paragraph), where the two classes differ only by a shift in the location. Let ${\bf X}$ and ${\bf Z}$ denote random elements from the two class distributions, where ${\bf Z}$ has the same distribution as ${\bf X} + {\bf c}$ for some fixed ${\bf c} \in {\cal X}$. Let us denote by $HD_{{\bf X}}$ and $HD_{{\bf Z}}$ the half-space depth functions computed using the distributions of ${\bf X}$ and ${\bf Z}$, respectively. Similarly, let $PD_{{\bf X}}$ and $PD_{{\bf Z}}$ be the projection depth functions based on the distributions of ${\bf X}$ and ${\bf Z}$, respectively. Then, under the assumptions of Theorem \ref{thm2.1} or \ref{thm2.2}, it is easy to verify using the arguments in the proofs of those theorems (see Section \ref{sec:5}) that $HD_{{\bf X}}({\bf w}) = HD_{{\bf Z}}({\bf w}) = PD_{{\bf X}}({\bf w}) = PD_{{\bf Z}}({\bf w}) = 0$  for almost every realization ${\bf w}$ of ${\bf X}$ and ${\bf Z}$. This implies that neither HD nor PD is suitable for classification purpose in the space ${\cal X}$ for such class distributions.   \\
\indent A modified version of Tukey depth, called the random Tukey depth (RTD), was proposed in \cite{CANR08} for probability distributions in $l_{2}$. It is defined as $RTD({\bf x}) = \min_{1 \leq j \leq N} \min\{P(\langle{\bf U}_{j},{\bf X}\rangle \leq \langle{\bf U}_{j},{\bf x}\rangle),P(\langle{\bf U}_{j},{\bf X}\rangle \geq \langle{\bf U}_{j},{\bf x}\rangle)\}$, where ${\bf U}_{j}$'s are $N$ i.i.d. observations from some probability distribution in $l_{2}$ independent of ${\bf X}$, and the probability in the definition of RTD is conditional on them. It is easy to see that the support of the distribution of $RTD(\widetilde{{\bf X}})$ is the whole of $[0,1/2]$ for Gaussian and many other distributions in $l_{2}$, where $\widetilde{{\bf X}}$ denotes an independent copy of ${\bf X}$.  However, \cite{CANR08} mentioned some theoretical and practical difficulties with RTD including the problem of choosing $N$ and the distribution of ${\bf U}_{j}$'s. A depth function for probability distributions in Banach 
spaces was introduced in \cite{CF09}, which is called Integrated dual depth (IDD). It is defined as $IDD({\bf x}) = \int_{{\cal X}^{*}} D_{{\bf u}}({\bf u}({\bf x}))Q(d{\bf u})$, where ${\bf x} \in {\cal X}$, $Q$ is a probability measure in ${\cal X}^{*}$, and $D_{{\bf u}}$ is a depth function defined on $\mathbb{R}$. \cite{CF09} recommended that one can choose a finite number of i.i.d. random elements ${\bf U}_{1}, {\bf U}_{2}, \ldots, {\bf U}_{N}$ from a probability distribution in ${\cal X}^{*}$, which will be independent of ${\bf X}$ and compute IDD using $N^{-1} \sum_{k=1}^{N} D_{{\bf U}_k}({\bf U}_{k}({\bf x}))$. It can be easily shown that if $D_{{\bf u}}$ is any standard depth function (e.g., HD, SD or simplicial depth) that maps $\mathbb{R}$ onto a non-degenerate interval, then for Gaussian and many other distributions of ${\bf X}$ in ${\cal X}$, $IDD({\bf X})$ will have a non-degenerate distribution with an appropriate interval as its support. However, like RTD, there are no natural guidelines 
available in practice for choosing the probability distribution $Q$ in the dual space ${\cal X}^{*}$ and the number $N$ of the random directions ${\bf U}_{j}$'s.

\section{Depths based on coordinate random variables} 
\label{sec:2} In this section, we shall discuss depths that use the underlying coordinate system of the sample space. We begin by considering BD and HRD that were discussed in the Introduction. BD and HRD of any ${\bf x} = \{x_{t}\}_{t \in [0,1]} \in C[0,1]$ with respect to the probability distribution of a random element ${\bf X} = \{X_{t}\}_{t \in [0,1]}$ in $C[0,1]$ are defined as 
\begin{eqnarray}
 BD({\bf x}) &=& \sum_{j=2}^{J} P\left(\min_{i=1,\ldots,j} X_{i,t} \leq x_{t} \leq \max_{i=1,\ldots,j} X_{i,t}, \ \forall \ t \in [0,1]\right) \ \ \mbox{and}   \label{s3.1}     \\
 HRD({\bf x}) &=& \min\{P(X_{t} \leq x_{t}, \ \forall \ t \in [0,1]), \ P(X_{t} \geq x_{t}, \ \forall \ t \in [0,1])\},  \label{s3.2}
\end{eqnarray}
respectively. Here ${\bf X}_{i} = \{X_{i,t}\}_{t \in [0,1]}$, $i=1,2,\ldots,J$, denote independent copies of ${\bf X}$. \cite{LPR09, LPR11} defined finite dimensional versions of these two depth functions as follows. For $J$ independent copies ${\bf X}_{i} = (X_{i,1},X_{i,2},\ldots,X_{i,d})$, $i=1,2,\ldots,J$, of ${\bf X} = (X_{1},X_{2},\ldots,X_{d})$ and a fixed ${\bf x} = (x_{1},x_{2},\ldots,x_{d})$, 
\begin{eqnarray*}
 BD({\bf x}) &=& \sum_{j=2}^{J} P\left(\min_{1 \leq i \leq j} X_{i,k} \leq x_{k} \leq \max_{1 \leq i \leq j} X_{i,k}, \ \forall \ k=1,2\ldots,d\right) \ \ \mbox{and}   \\
 HRD({\bf x}) &=& \min\{P(X_{k} \leq x_{k}, \ \forall \ \ k=1,2\ldots,d), \ P(X_{k} \geq x_{k}, \ \forall \ \ k=1,2\ldots,d)\},
\end{eqnarray*}
respectively. The above definitions of BD and HRD in function spaces and finite dimensional Euclidean spaces lead to a natural definition of these depth functions in a sequence space. For $J$ i.i.d. copies ${\bf X}_{i} = (X_{i,1},X_{i,2},\ldots)$ of an infinite random sequence ${\bf X} = (X_{1},X_{2},\ldots)$ and a fixed sequence ${\bf x} = (x_{1},x_{2},\ldots)$, we can define 
\begin{eqnarray*}
BD({\bf x}) &=& \sum_{j=2}^{J} P\left(\min_{1 \leq i \leq j} X_{i,k} \leq x_{k} \leq \max_{1 \leq i \leq j} X_{i,k}, \ \forall \ k \geq 1\right) \ \ \mbox{and} \\
HRD({\bf x}) &=& \min\{P(X_{k} \leq x_{k}, \ \forall \ k \geq 1), \ P(X_{k} \geq x_{k}, \ \forall \ k \geq 1)\},
\end{eqnarray*}
respectively. However, as the following theorem shows, such versions of BD and HRD in sequence spaces will have degenerate behaviour for certain $\alpha$-mixing sequences.
\vspace{-0.05in}
\begin{theorem} \label{thm3.1}
Let ${\bf X} = (X_{1},X_{2},\ldots)$ be an $\alpha$-mixing sequence of random variables and denote the distribution of ${\bf X}$ by $\mu$. Also, assume that the mixing coefficients $\{\alpha_{k}\}$ satisfy $\sum_{k=1}^{\infty} \alpha_{k}^{1-1/2p} < \infty$ for some $p \geq 1$, and the $X_{k}$'s are non-atomic for each $k \geq 1$. Then, $BD({\bf x}) = HRD({\bf x}) = 0$ for all ${\bf x}$ with $\mu$-measure one, where $BD({\bf x})$ and $HRD({\bf x})$ denote the band and the half-region depths of ${\bf x}$ with respect to $\mu$, respectively.
\end{theorem}
\vspace{-0.05in}
\indent The preceding theorem implies that for i.i.d. copies of a random sequence satisfying appropriate $\alpha$-mixing conditions, any given sample sequence will not lie in a band or a half-region formed by the other sample sequences with probability one. A question that now arises is whether a similar phenomenon holds for probability distributions in function spaces like $C[0,1]$. Unfortunately, as the next theorem shows, BD and HRD continue to exhibit degenerate behaviour for a well-known class of probability measures in $C[0,1]$. 
\vspace{-0.05in}
\begin{theorem} \label{thm3.2}
Let $\{X_{t}\}_{t \in [0,1]}$ be a Feller process having continuous sample paths. Assume that for some $x_{0} \in \mathbb{R}$, $P(X_{0} = x_{0}) = 1$, and the distribution of $X_{t}$ is non-atomic and symmetric about $x_{0}$ for each $t \in (0,1]$. Then, $BD({\bf x}) = HRD({\bf x}) = 0$ for all ${\bf x}$ in a set of $\mu$-measure one, where $\mu$ denotes the probability distribution of ${\bf X}$, and the depth functions BD and HRD are obtained using $\mu$.
\end{theorem}
\vspace{-0.05in}
We refer to \cite{RY91} for an exposition on Feller processes that include Brownian motions, geometric Brownian motions and Brownian bridges. The above theorem implies that for many well-known stochastic processes, BD and HRD will be degenerate at zero. Consequently, BD and HRD will not be suitable for depth-based statistical procedures like trimming, identification of central and outlying data points, etc. for such distributions in $C[0,1]$ like HD and PD. Consider next distinct Feller processes ${\bf X}$ and ${\bf Y}$ on $C[0,1]$, and let $BD_{{\bf X}}$, $BD_{{\bf Y}}$, $HRD_{{\bf X}}$ and $HRD_{{\bf Y}}$ denote the BD's and the HRD's obtained using the distributions of ${\bf X}$ and ${\bf Y}$, respectively. Then, if both of ${\bf X}$ and ${\bf Y}$ satisfy the conditions of Theorem \ref{thm3.2}, using the arguments in the proofs of Lemma \ref{lem6.4} and \ref{lem6.5} (see Section \ref{sec:5}), it follows that $BD_{{\bf X}}({\bf z}) = BD_{{\bf Y}}({\bf z}) = HRD_{{\bf X}}({\bf z}) = HRD_{{\bf Y}}({\bf z}) = 
0$ for almost every realization ${\bf z}$ of ${\bf X}$ and ${\bf Y}$. This implies that neither BD nor HRD will be able to discriminate between the distributions of ${\bf X}$ and ${\bf Y}$.  \\
\indent As mentioned in the Introduction, it was observed by \cite{LPR09, LPR11} that the depth functions BD and HRD tend to take small values if the sample curves intersect each other often; and this observation motivated them to consider modified versions of BD and HRD, namely MBD and MHRD, respectively. MBD and MHRD for probability distributions in $C[0,1]$, as defined by \cite{LPR09, LPR11}, are given below. For a fixed ${\bf x} = \{x_{t}\}_{t \in [0,1]} \in C[0,1]$ and $J$ i.i.d. copies ${\bf X}_{i} = \{X_{i,t}\}_{t \in [0,1]}$ of a random element ${\bf X} = \{X_{t}\}_{t \in [0,1]} \in C[0,1]$,
\begin{eqnarray*} 
MBD({\bf x}) &=&  \sum_{j=2}^{J} E\left[\lambda\left(\left\{t \in [0,1] : \min_{i=1,\ldots,j} X_{it} \leq x_{t} \leq \max_{i=1,\ldots,j} X_{it}\right\}\right)\right] \ \mbox{and} \\
MHRD({\bf x}) &=&  \min\{E[\lambda(\{t \in [0,1] : X_{t} \leq x_{t}\})], \ E[\lambda(\{t \in [0,1] : X_{t} \geq x_{t}\})]\}, 
\end{eqnarray*}
where $\lambda(.)$ is the Lebesgue measure on $[0,1]$. \cite{FM01} defined ID for probability measures on $C[0,1]$ as follows. For ${\bf x} = \{x_{t}\}_{t \in [a,b]} \in C[0,1]$ and a random element ${\bf X} = \{X_{t}\}_{t \in [0,1]} \in C[0,1]$, $ID({\bf x}) = \int_{0}^{1} D_{t}(x_{t})dt$, where for every $t$, $D_{t}(.)$ denotes a univariate depth function on the real line obtained using the distribution of $X_{t}$. As observed in \cite{LPR09}, if we choose $J = 2$ in the definition of MBD, then $MBD({\bf x}) = \int_{0}^{1} 2F_{t}(x_{t})(1 - F_{t}(x_{t}))dt$, which is $ID({\bf x})$ defined using the simplicial depth for each coordinate variable. Here $F_{t}$ denotes the distribution of $X_{t}$ for each $t \in [0,1]$. Indeed, we have the following equivalent representations of MBD and MHRD by Fubini's theorem. For any ${\bf x} = \{x_{t}\}_{t \in [0,1]} \in C[0,1]$, 
\begin{eqnarray}
MBD({\bf x}) &=& \sum_{j=2}^{J} E\left[\int_{0}^{1} I\left(\min_{i=1,\ldots,j} X_{it} \leq x_{t} \leq \max_{i=1,\ldots,j} X_{it}\right) dt\right] \nonumber \\
&=& \sum_{j=2}^{J} \int_{0}^{1} \left[1 - F_{t}^{j}(x_{t}-) - (1 - F_{t}(x_{t}))^{j}\right] dt \ \ \mbox{and}    \label{s3.3}  \\
MHRD({\bf x}) &=& \min\left\{E\left[\int_{0}^{1} I(X_{t} \leq x_{t}) dt\right], \ E\left[\int_{0}^{1} I(X_{t} \geq x_{t}) dt\right]\right\}  \nonumber \\
&=& \min\left\{\int_{0}^{1} P(X_{t} \leq x_{t})dt, \ \int_{0}^{1} P(X_{t} \geq x_{t})dt\right\}.    \label{s3.4}
\end{eqnarray}
\indent It is easy to see from (\ref{s3.3}) that if ${\bf X} = \{X_{t}\}_{t \in [0,1]} \in C[0,1]$ is symmetrically distributed about ${\bf a} = \{a_{t}\}_{t \in [0,1]} \in C[0,1]$, i.e., ${\bf X} - {\bf a}$ and ${\bf a} - {\bf X}$ have the same distribution, then MBD has a unique maximum at ${\bf a}$. The same is true for ID provided that for all $t \in [0,1]$, the univariate depth $D_{t}$ in the definition of ID has a unique maximum at $a_{t}$ (cf. the property ``FD4center'' in \citet[p. $10$]{MP12}, Theorems $3$ and $4$ in \cite{Liu90} and property ``P2'' in \citet[p. $463$]{ZS00a}). Consider next ${\bf x} = \{x_{t}\}_{t \in [0,1]} \in C[0,1]$ and ${\bf y} = \{y_{t}\}_{t \in [0,1]} \in C[0,1]$ satisfying either $a_{t} \leq x_{t} \leq y_{t}$ or $y_{t} \leq x_{t} \leq a_{t}$ for all $t \in [0,1]$, i.e., ${\bf y}$ is farther away from ${\bf a}$ than ${\bf x}$. Then, $MHRD({\bf y}) \leq MHRD({\bf x})$ and $MBD({\bf y}) \leq MBD({\bf x})$. Further, if $D_{t}(x_{t})$  is a decreasing function of $|x_{t} - a_{t}|$ for all $t \in [0,1]$, we have $ID({\bf y}) \leq ID({\bf x})$ (cf. the ``FD4pw Monotone'' property in \citet[p. $9$]{MP12}). Consider next any ${\bf x} = \{x_{t}\}_{t \in [0,1]} \in C[0,1]$ satisfying $x_{t} \neq 0$ for all $t$ in a subset of $[0,1]$ with Lebesgue measure one. It follows from representations (\ref{s3.3}) and (\ref{s3.4}) for MBD and MHRD that both $MBD({\bf a} + n{\bf x})$ and $MHRD({\bf a} + n{\bf x})$ converge to zero as $n \rightarrow \infty$. Further, if $D_{t}(s) \rightarrow 0$ as $|s - a_{t}| \rightarrow \infty$ for all $t \in [0,1]$, then $ID({\bf a} + n{\bf x}) \rightarrow 0$ as $n \rightarrow \infty$. So, all these depth functions tend to zero as one moves away from the center of symmetry along suitable lines. This can be viewed as a weaker version of the ``FD3'' property in \cite{MP12} (see also Theorem $1$ in \cite{Liu90} and property ``P4'' in \citet[p. $464$]{ZS00a}). \\
\indent The following theorem shows that MBD, MHRD and ID have non-degenerate distributions with adequate spread for a class of probability distributions in $C[0,1]$ that includes many popular stochastic models. The properties of these depth functions discussed above and the theorem stated below show that these depth functions are suitable choices for a center-outward ordering of elements of $C[0,1]$ with respect to the distributions of a large class of stochastic processes, and can be used for constructing central and outlying regions, trimmed estimators, and also for outlier detection. Moreover, due to the continuity of ID and MBD, and the fact that they attain their unique maximum at the centre of symmetry of any probability distribution, both of these depth functions will be able to discriminate between two distributions with distinct centres of symmetry. \\
\indent For the next theorem, in the definition of ID, we shall assume $D_{t}(.) = \psi(F_{t}(.))$ for all $t \in [0,1]$, where $\psi$ is a bounded continuous positive function satisfying $\psi(0+) = \psi(1-) = 0$, and $F_{t}$ denotes the distribution of $Y_{t}$.
\vspace{-0.05in}
\begin{theorem} \label{thm3.3}
Consider the process ${\bf X} = \{X_{t}\}_{t \in [0,1]} = \{g(t,Y_{t})\}_{t \in [0,1]}$, where $\{Y_{t}\}_{t \in [0,1]} \in C[0,1]$ is a fractional Brownian motion starting at some $y_{0} \in \mathbb{R}$. Assume that the function $g : [0,1] \times \mathbb{R}$ is continuous, and $g(t,.)$ is strictly increasing with $g(t,s) \rightarrow \infty$ as $s \rightarrow \infty$ for each $t \in [0,1]$. Then the following hold. \\
(a) The depth functions $MBD({\bf x})$, $MHRD({\bf x})$ and $ID({\bf x})$ take all values in $(0,A_{J}]$, $(0,1/2]$ and $\psi((0,1))$, respectively, as ${\bf x}$ varies in $C[0,1]$, where MBD, MHRD and ID are obtained using the distribution of ${\bf X}$, and $A_{J} = J - 2 + 2^{-J+1}$ for any $J \geq 2$ with $J$ as in the definitions of BD and MBD. \\
(b) The supports of the distributions of $MBD(\widetilde{{\bf X}})$, $MHRD(\widetilde{{\bf X}})$ and $ID(\widetilde{{\bf X}})$ are $[0,A_{J}]$, $[0,1/2]$ and the closure of $\psi((0,1))$, respectively. Here $\widetilde{{\bf X}}$ denotes an independent copy of ${\bf X}$.  \\
(c) The above conclusions hold if $\{Y_{t}\}_{t \in [0,1]}$ is a fractional Brownian bridge starting at $y_{0} \in \mathbb{R}$.
\end{theorem}
\vspace{-0.05in}
\indent Note that since $\psi$ is a continuous non-constant function, the support of the distribution of $ID(\widetilde{{\bf X}})$ is actually a closed non-degenerate interval. Here, by the support of a probability distribution in any metric space, we mean the smallest closed set with probability one. Let us also observe that in the above theorem, the depths are computed based on the entire process ${\bf X} = \{X_{t}\}_{t \in [0,1]}$ starting from time $t = 0$. But in practice, it might very often be the case that we observe the process from some time point $t_{0} > 0$, and then the depths are to be computed based on the observed path $\{X_{t}\}_{t \in [t_{0},1]}$. Even in that case, the conclusions of the above theorem hold (see Remark \ref{rem6.1} in Section \ref{sec:5}). 

\section{Spatial depth in infinite dimensional spaces}
\label{sec:3} In this section, we shall consider an extension of the notion of spatial depth from $\mathbb{R}^{d}$ into infinite dimensional spaces. Spatial depth of ${\bf x} \in \mathbb{R}^{d}$ with respect to the probability distribution of a random vector ${\bf X} \in \mathbb{R}^{d}$ is defined as $SD({\bf x}) = 1 - ||E\{({\bf x} - {\bf X})/||{\bf x} - {\bf X}||\}||$ (see, e.g., \cite{VZ00} and \cite{Serf02}). It has been widely used for various statistical procedures including clustering and classification (see, e.g., \cite{Jorn04} and \cite{GC05a}), construction of depth-based central and outlying regions and depth-based trimming (see \cite{Serf06}). This depth function extends naturally to any Hilbert space ${\cal X}$. For an ${\bf x} \in {\cal X}$ and a random element ${\bf X} \in {\cal X}$, we can define $SD({\bf x})$ using the same expression as above, where $||.||$ is to be taken as the usual norm in ${\cal X}$, and the expectation is in the Bochner sense (see, e.g., \citet[p. $100$]{AG80}). 
Spatial depth function inherits many of its interesting properties from finite dimensions. For instance, $SD({\bf x})$ is a continuous function in ${\bf x}$ if the distribution of ${\bf X}$ is non-atomic, which is a direct consequence of dominated convergence theorem. Further, it follows from \cite{Kemp87} that if ${\cal X}$ is strictly convex, and the distribution of ${\bf X}$ is non-atomic and is not supported on a line in ${\cal X}$, then the function SD has a unique maximum at the spatial median ${\bf m}$ (say) of ${\bf X}$, and its maximum value is $1$ (cf. the property ``FD4center'' in \citet[p. $10$]{MP12}, Theorems $3$ and $4$ in \cite{Liu90} and property ``P2'' in \citet[p. $463$]{ZS00a}). Further, if we consider the sequence $\{{\bf m} + n{\bf x}\}_{n \geq 1}$ for any fixed non-zero ${\bf x} \in {\cal X}$, it follows by a simple application of dominated convergence theorem that $SD({\bf m} + n{\bf x}) \rightarrow 0$ as $n \rightarrow \infty$ (cf. the ``FD3'' property in \cite{MP12}, Theorem $1$ in \cite{Liu90} and property ``P4'' in \citet[p. $464$]{ZS00a}). \\
\indent A natural question that arises now is whether SD suffers from degeneracy similar to what was observed in the case of some of the depth functions discussed earlier or its distribution is well spread out. As the next theorem shows, the distribution of SD is actually supported on the entire unit interval for a large class of probability measures in a separable Hilbert space ${\cal X}$ including Gaussian probabilities.
\vspace{-0.05in}
\begin{theorem} \label{thm4.1}
Let ${\cal X}$ be a separable Hilbert space and consider a random element ${\bf X} = \sum_{k=1}^{\infty} X_{k}\phi_{k}$, where $\{\phi_{k}\}_{k \geq 1}$ is an orthonormal basis of ${\cal X}$. Assume that ${\bf X}$ has a nonatomic probability distribution $\mu$ with $\sum_{k=1}^{\infty} E(X_{k}^{2}) < \infty$, and the support of the conditional distribution of $(X_{1},X_{2},\ldots,X_{d})$ given $(X_{d+1},X_{d+2},\ldots)$ is the whole of $\mathbb{R}^{d}$ for each $d \geq 1$. Then, the function $SD({\bf x})$ defined using the distribution $\mu$ takes all the values in $(0,1]$ as ${\bf x}$ varies in ${\cal X}$. Further, if $\widetilde{{\bf X}}$ denotes an independent copy of ${\bf X}$, the support of the distribution of $SD(\widetilde{{\bf X}})$ will be the whole of $[0,1]$.
\end{theorem}
\vspace{-0.05in}
\indent Since $C[0,1] \subseteq L_{2}[0,1]$, for any probability distribution on $C[0,1]$, SD is defined in the same way as in the case of the separable Hilbert space $L_{2}[0,1]$. Thus, for a random element ${\bf X} \in C[0,1]$, if the sequence $(X_{1},X_{2},\ldots)$ obtained from the orthogonal decomposition of ${\bf X}$ in $L_{2}[0,1]$ satisfies the conditions of Theorem \ref{thm4.1}, then the support of the distribution of $SD(\widetilde{{\bf X}})$ will be the whole of $[0,1]$. In particular, for any Gaussian process having a continuous mean function and a continuous positive definite covariance kernel, we can have $(X_{1},X_{2},\ldots)$ to be the coefficients of the Karhunen-Loeve expansion of ${\bf X}$, which will then be a sequence of independent Gaussian random variables, and consequently, the conditions of Theorem \ref{thm4.1} will hold. Those assumptions, however, need not hold when ${\bf X}$ is a function of some Gaussian process in $C[0,1]$ like what we have considered in Theorem \ref{thm3.3}. 
Indeed, even if ${\bf X}$ admits a Karhunen-Loeve type expansion in such a case, the sequence of coefficients need not satisfy the conditions of Theorem \ref{thm4.1}. However, as the next theorem shows, the distribution of SD has full support on the unit interval in some of these situations as well. 
\vspace{-0.05in}
\begin{theorem}  \label{thm4.2}
Consider the process ${\bf X} = \{X_{t}\}_{t \in [0,1]} = \{g(t,Y_{t})\}_{t \in [0,1]}$ as in Theorem \ref{thm3.3}. Then, the function $SD({\bf x})$ defined using the distribution of ${\bf X}$ takes all values in $(0,1)$ as ${\bf x}$ varies in $C[0,1]$. Moreover, the support of the distribution of $SD(\widetilde{{\bf X}})$ is the whole of $[0,1]$, where $\widetilde{{\bf X}}$ is an independent copy of ${\bf X}$. 
\end{theorem}
\vspace{-0.05in}
\indent It follows from arguments that are very similar to those in Remark \ref{rem6.1} in Section \ref{sec:5} that the above result holds even if SD is computed based on the process $\{X_{t}\}_{t \in [t_{0},1]}$, where $t_{0} > 0$. The properties of SD stated at the beginning of this section along with the results in Theorems \ref{thm4.1} and \ref{thm4.2} imply that like ID, MBD and MHRD, SD can also be used for various depth-based statistical procedures. The spatial depth function can also be used to discriminate between two probability measures in a separable Hilbert space or $C[0,1]$. For instance, for any two non-atomic probability measures having distinct and unique spatial medians, the associated spatial depth functions will be continuous, each having a unique maximum at the corresponding spatial median. In that case, spatial depth will be able to distinguish between the two distributions. \\
\indent We conclude this section with the discussion of another notion of depth, which has a somewhat similar nature as that of SD. For a random element ${\bf X}$ and a fixed element ${\bf x}$ in $L_{2}[0,1]$, the h-depth introduced in \cite{CFF07} is defined as $E\{K_{h}(||{\bf x} - {\bf X}||)\}$, where $K_{h}(t) = h^{-1}K(t/h)$ for some fixed kernel $K$ and $h > 0$ is a tuning parameter. Suppose that $K$ is a bounded continuous probability density function supported on the whole of $[0,\infty)$ with $K(s) \rightarrow 0$ as $s \rightarrow \infty$, and $(X_{1},X_{2},\ldots)$ satisfies the conditions in Theorem \ref{thm4.1}, which ensures that the support of the distribution of ${\bf X}$ is the whole of $L_{2}[0,1]$. Then, in view of the continuity of $E\{K_{h}(||{\bf x} - {\bf X}||)\}$ as a function of ${\bf x}$, which is a consequence of the dominated convergence theorem, it is not difficult to see that the support of the distribution of the h-depth evaluated at an independent copy $\widetilde{{\bf X}}$ of 
${\bf X}$ will be the whole of $[0,A]$. Here $A = \sup_{{\bf x} \in L_{2}[0,1]} E\{K_{h}(||{\bf x} - {\bf X}||)\}$. However, no specific guidelines are available for choosing $K$ and $h$ in practice.

\section{Demonstration using real and simulated data}
\label{sec:4} In the three preceding sections, we have investigated the behaviour of several depth functions in infinite dimensional spaces. The results derived in those sections are all about the population versions of different depth functions. In this section, we try to investigate to what extent those results are reflected in the empirical versions of the corresponding depth functions computed using some simulated and real datasets. First, we shall consider some simulated and real sequence data. The simulated dataset consists of $50$ i.i.d. observations on a $d$-dimensional Gaussian random vector ${\bf X} = (X_{1},X_{2},\ldots,X_{d})$ that satisfies $Cov(X_{k},X_{l}) =  r^{-|k-l|}/(kl)^{2}$, where $r = 0.1$, $k,l = 1,2,\ldots,d$, and $d=4000$. The real dataset that is considered next is obtained from \texttt{http://datam.i2r.a-star.edu.sg/datasets/krbd/ColonTumor/ColonTumor.zip}, and it contains expressions of $d=2000$ genes in tumor tissue biopsies corresponding to $40$ colon tumor patients and $22$ normal samples of colon tissue. For both these datasets, we can view each sample point as the first $d$ coordinates of an infinite sequence. \\
\indent In all our samples, since the dimension is much larger than the sample size, the empirical versions of both of HD and PD turn out to be zero (see Figure \ref{Fig:1}). This is a consequence of the fact that when the dimension is larger than the sample size, and no sample point lies in the subspace spanned by the remaining sample points, the HD and the PD of any data point with respect to the empirical distribution of the remaining data points is zero (see, e.g., remarks at the beginning of Section $4$ in \cite{DGC11}). It is also observed from the dotplots in Figure \ref{Fig:1} that empirical BD and HRD are both degenerate at zero for the two datasets. However, the distribution of empirical SD is well spread out in the corresponding dotplots in Figure \ref{Fig:1}. \\
\indent For the colon data, we have prepared another dotplot (see Fig. \ref{Fig:4}), which shows the difference between the two empirical SD values for each data point, where one depth value is obtained with respect to the empirical distribution of the tumor tissue sample, and the other one is obtained using that of the normal sample. The value of this difference for a data point corresponding to the tumor tissue is plotted in the panel with heading ``Tumor tissue'', where all the values are positive. This implies that each data point in the sample of tumor tissue has higher depth value with respect to the empirical distribution of the tumor sample than its depth value with respect to the empirical distribution of the normal tissue sample. On the other hand, a data point corresponding to the normal tissue is plotted in the panel with heading ``Normal tissue'', where all the values, except only two, are negative. In other words,  except for those two cases, each data point in the sample of normal tissue has higher depth value with respect to the normal tissue sample. Thus, SD adequately discriminates between the two samples, and maximum depth or other depth-based classifiers (see, e.g., \cite{GC05a} and \cite{LCAL12}) constructed using SD will yield good results for this dataset.

\begin{figure}[htbp]
\begin{center}
\includegraphics[height=3.8in,width=3.0in,angle=270]{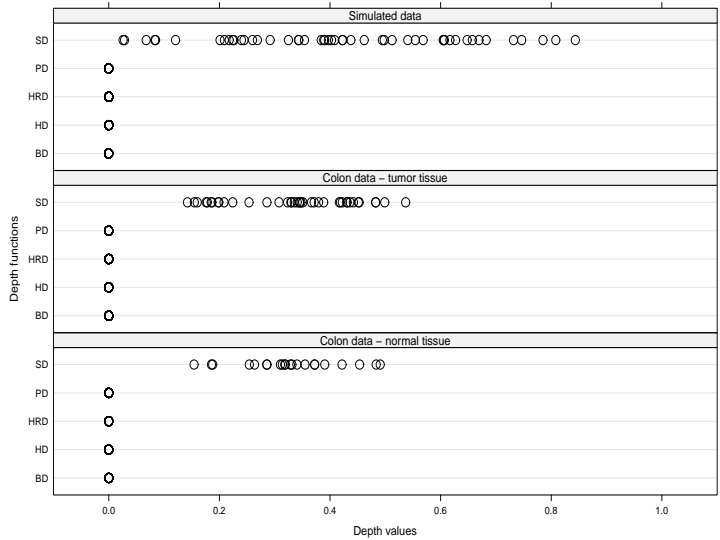} 
\end{center}
\caption{{\it Dotplots of SD, PD, HRD, HD and BD for simulated data and colon data.}}
\label{Fig:1}
\end{figure}

\begin{figure}[htbp]
\begin{center}
\includegraphics[height=3.7in,width=2.7in,angle=270]{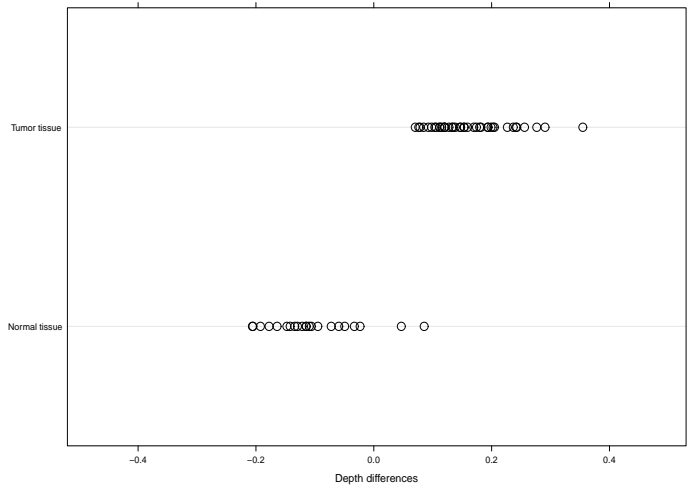} 
\end{center}
\caption{{\it Dotplots of depth differences based on SD for colon data. The horizantal axis corresponds to the difference between empirical SD values of each data point with respect to the tumor tissue sample and the normal tissue sample.}}
\label{Fig:4}
\end{figure}

\indent We shall next consider some simulated and real functional data. Each of the three simulated datasets consists of $50$ observations from (i) a standard Brownian motion on $[0,1]$, (ii) a zero mean fractional Brownian motion on $[0,1]$ with covariance function $K(t,s) = (1/2) [t^{2H} + s^{2H} - |t-s|^{2H}]$, where $t,s \in [0,1]$, and we choose the Hurst index $H = 0.75$, and (iii) a geometric Brownian motion defined as $X_{t} = exp((r-\sigma^2/2)t + \sigma B_{t})$, where $t \in [0,1]$ and $r = \sigma = 0.5$. Here $\{B_{t}\}_{t \in [0,1]}$ denotes the standard Brownian motion on $[0,1]$. For all three simulated datasets, the sample functions were observed at $d=2000$ equispaced points in $(0,1)$. We have also considered two real datasets, the first one being the lip movement data, which is available at \texttt{www.stats.ox.ac.uk/$\sim$silverma/fdacasebook/LipPos.dat}, and contains $32$ sample observations on the movement of the lower lip. The curves are the trajectories traced by the lower lip while pronouncing the word ``bob''. The measurements are taken at $d=501$ time points in a time interval of $700$ milliseconds. The second real dataset is the growth acceleration dataset derived from the well-known Berkeley growth data (see \cite{RS05}), which contains two subclasses, namely, the boys and the girls. Heights of $39$ boys and $54$ girls were measured at $31$ time points between ages $1$ and $18$ years. The growth acceleration curves are obtained through monotone spline smoothing available in the R package ``fda'', and those are recorded at $d=101$ equispaced ages in the interval $[1,18]$. For these functional datasets, we calculated MBD by taking $J = 2$ as suggested in \cite{LPR09}, and $D_{t}$ in the definition of ID was taken to be SD for each $t$, which is equivalent to the depth function used in \cite{FM01}. \\
\indent As shown in the dotplots in Figures \ref{Fig:2} and \ref{Fig:3}, for all of the above simulated and real data, the distributions of empirical ID, MBD, MHRD and SD are well spread out. Empirical BD and HRD are both degenerate at zero for the Brownian motion and the geometric Brownian motion (see Figure \ref{Fig:2}). For the fractional Brownian motion, the maximum value of empirical BD was $0.024$, with its median $= 0$ and the third quartile $= 0.004$, whereas the maximum value of empirical HRD was $0.020$ with its third quartile $= 0$ (see Figure \ref{Fig:2}). For the lip movement data, the empirical HRD is degenerate at zero, while the maximum value of empirical BD is $0.006$ with its third quartile $= 0$ (see Figure \ref{Fig:3}). For the growth acceleration data, the HRD again turns out to be degenerate at zero, while BD takes a maximum value of $0.004$ for boys and $0.008$ for girls, and the third quartile for BD $= 0$ for boys as well as girls (see Figure \ref{Fig:3}). \\
\indent For the growth acceleration data, Fig. \ref{Fig:5} shows the dotplots for the differences between the two depth values with respect to the empirical distributions of the boys and the girls based on SD, MHRD, MBD and ID. The value of this difference for a data point corresponding to a boy (respectively, a girl) is plotted in the panel with heading ``Boys'' (respectively, ``Girls''). For SD, MBD and ID, most of the data points corresponding to the boys have higher depth values with respect to the empirical distribution of the boys than with respect to the empirical distribution of the girls. On the other hand, most of the data points corresponding to the girls have higher depth values with respect to the empirical distribution of the girls. This implies that each of ID, MBD and SD adequately discriminates between the two samples, and depth-based classifiers (see, e.g., \cite{GC05a} and \cite{LCAL12}) constructed using ID, MBD or SD will perform well for this dataset. However, the plot corresponding to 
MHRD shows that a large number of data points in the sample of boys have higher depth values with respect to the empirical distribution of the girls, and almost half of the data points in the sample of girls have higher depth values with respect to the empirical distribution of the boys. This indicates that MHRD does not discriminate well between the two samples.

\begin{figure}[htbp]
\begin{center}
\includegraphics[height=4.6in,width=3.6in,angle=270]{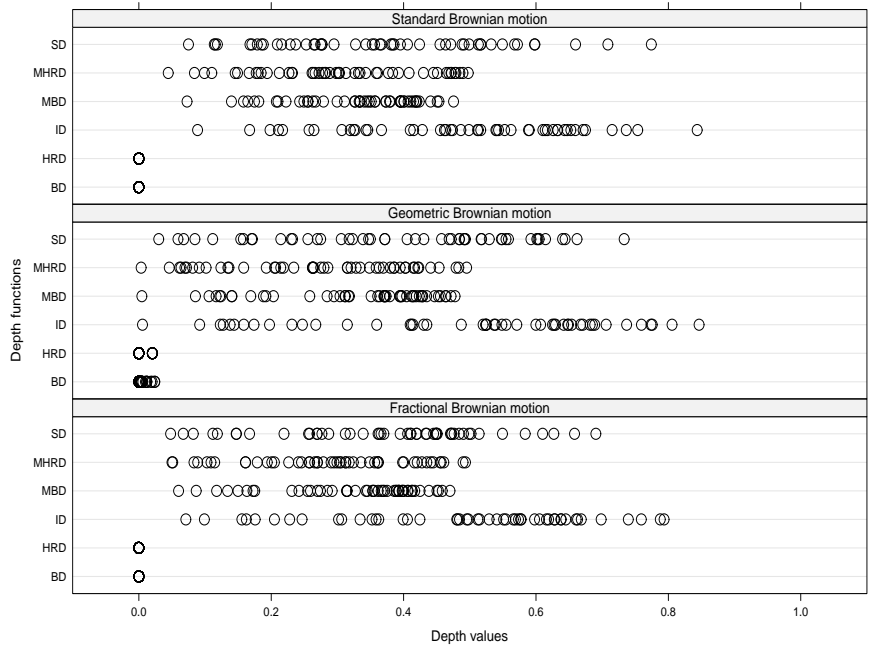} 
\end{center}
\caption{{\it Dotplots of SD, MHRD, MBD, ID, HRD and BD for simulated standard Brownian motion, geometric Brownian motion and fractional Brownian motion.}}
\label{Fig:2}
\end{figure}

\begin{figure}[htbp]
\begin{center}
\includegraphics[height=4.6in,width=3.6in,angle=270]{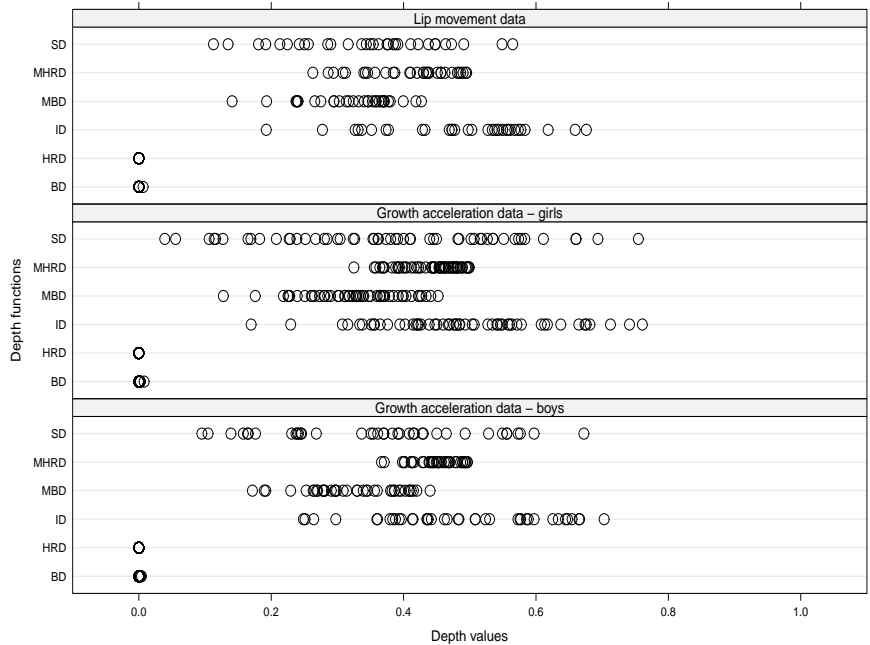} 
\end{center}
\caption{{\it Dotplots of SD, MHRD, MBD, ID, HRD and BD for lip movement data and growth acceleration data.}}
\label{Fig:3}
\end{figure}

\begin{figure}[htbp]
\begin{center}
\includegraphics[height=4.3in,width=3.1in,angle=270]{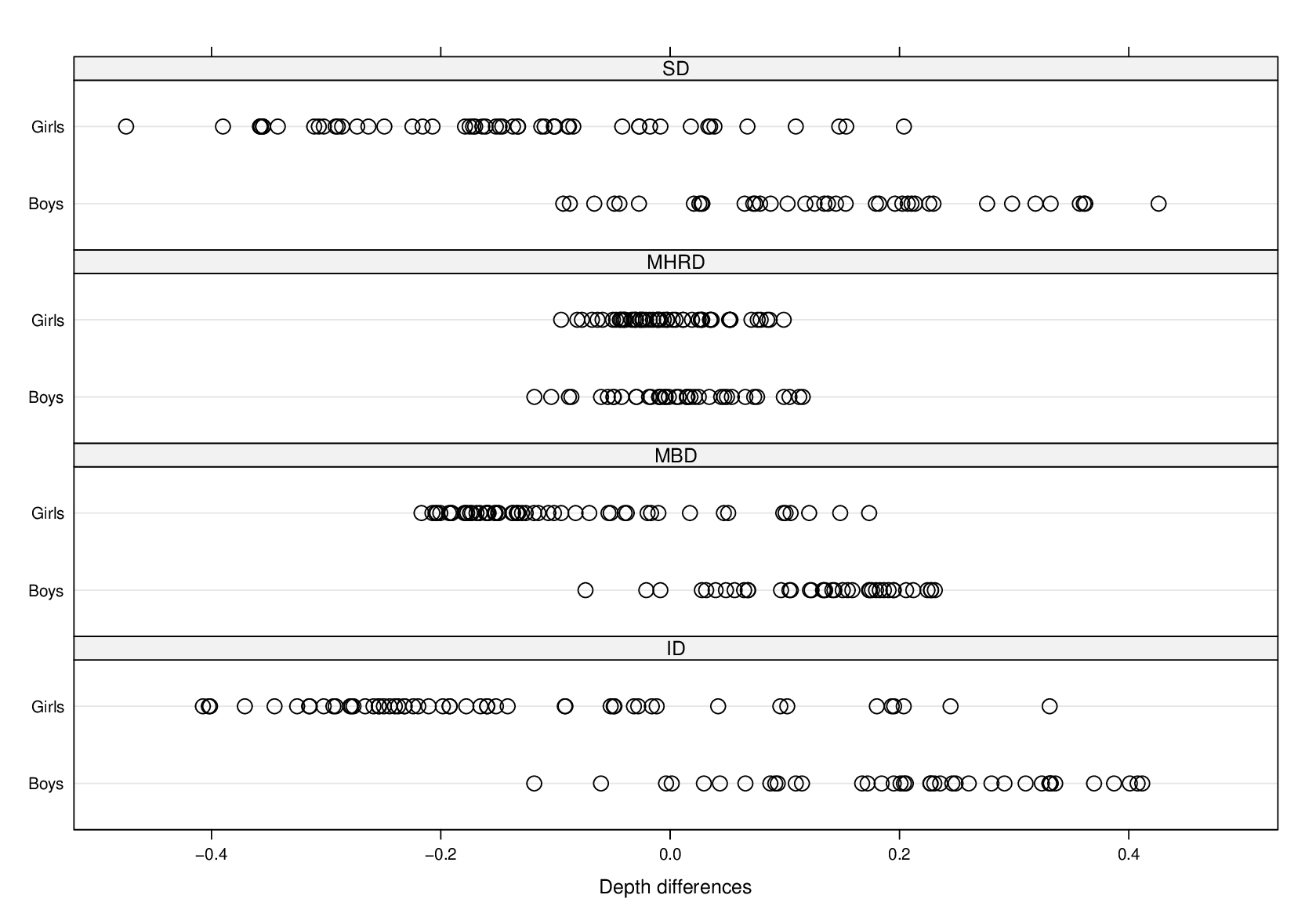} 
\end{center}
\caption{{\it Dotplots for depth differences based on SD, MHRD, MBD and ID for growth accceleration data. The horizantal axis corresponds to the difference  between empirical depth values of each data point with respect to the boys and the girls.}}
\label{Fig:5}
\end{figure}

\section{Technical details}
\label{sec:5} 
\begin{proof}\textit{of Theorem \ref{thm2.1}} \ \ 
Let ${\bf X}(d) = (X_{1},X_{2},\ldots,X_{d})'$ and ${\bf Y}(d) = (Y_{1},Y_{2},\ldots,Y_{d})'$ be $d$-dimensional column vectors that consist of the first $d$ coordinates of the sequences ${\bf X}$ and ${\bf Y}$. Observe that ${\bf Y}(d) = T_{d}({\bf X}(d))$, where $T_{d} : \mathbb{R}^{d} \rightarrow \mathbb{R}^{d}$ is a bijective affine map. By definition, the half-space depth of a point ${\bf x} \in l_{2}$ relative to the distribution of ${\bf X}$ will satisfy
\begin{eqnarray}
HD({\bf x}) &=& \inf_{{\bf u} \in l_{2}} P(\langle{\bf u},{\bf X} - {\bf x}\rangle \geq 0) \ \leq \ \inf_{d \geq 1} \inf_{{\bf v} \in \mathbb{R}^{d}} \ P({\bf v}'{\bf X}(d) \geq {\bf v}'{\bf x}(d)) \nonumber \\ 
&=& \inf_{d \geq 1} \inf_{{\bf v} \in \mathbb{R}^{d}} \ P({\bf v}'{\bf Y}(d) \geq {\bf v}'{\bf y}(d))   \nonumber \\
&\leq& \inf_{d \geq 1} \inf_{{\bf v} \in \mathbb{R}^{d} : {\bf v}'{\bf y}(d) > 0} \ P({\bf v}'{\bf Y}(d) \geq {\bf v}'{\bf y}(d)),   \label{eq2.1.1}
\end{eqnarray}
where ${\bf x}(d) = (x_{1},x_{2},\ldots,x_{d})'$ is the vector of first $d$ coordinates of ${\bf x}$, ${\bf y}(d) = (y_{1},y_{2},\ldots,y_{d})' \newline = T_{d}({\bf x}(d))$ and ${\bf v} = (v_{1},v_{2},\ldots,v_{d})'$. Throughout this section, any finite dimensional vector will be a column vector, and $'$ will denote its transpose. Since $Y_{1},Y_{2},\ldots,Y_{d}$ are uncorrelated, it follows from (\ref{eq2.1.1}) and Chebyshev inequality that 
\begin{eqnarray}
HD({\bf x}) &\leq&  \inf_{d \geq 1} \inf_{{\bf v} : {\bf v}'{\bf y}(d) > 0} \  \frac{Var({\bf v}'{\bf Y}(d))}{({\bf v}'{\bf y}(d))^{2}} \ = \ \inf_{d \geq 1} \inf_{{\bf v} : {\bf v}'{\bf y}(d) > 0} \ \frac{\sum_{k=1}^{d} v_{k}^{2}\tau_{k}^{2}}{\left[\sum_{k=1}^{d} v_{k}y_{k}\right]^{2}}.   \label{eq2.1.2}
\end{eqnarray}
(\ref{eq2.1.2}) implies, by an application of Cauchy-Schwarz inequality, that
\begin{eqnarray}
 HD({\bf x}) \leq \inf_{d \geq 1} \ \left[\sum_{k=1}^{d} y_{k}^{2}/\tau_{k}^{2}\right]^{-1}.   \label{eq2.1.3}
\end{eqnarray}
\indent In view of the moment and the mixing conditions assumed on the $Y_{k}$'s in the theorem, it follows from Corollary 4 in \cite{Hans91} that
\begin{eqnarray}
d^{-1} \sum_{k=1}^{d} Y_{k}^{2}/\tau_{k}^{2} \rightarrow 1 \ \mbox{a.s.} \ \Rightarrow \  \inf_{d \geq 1} \ \left[\sum_{k=1}^{d} Y_{k}^{2}/\tau_{k}^{2}\right]^{-1} = 0 \ \mbox{a.s.}    \label{eq2.1.4}
\end{eqnarray}
(\ref{eq2.1.3}) and (\ref{eq2.1.4}) imply that $HD({\bf x}) = 0$ for all ${\bf x}$ in a subset of $l_{2}$ with $\mu$-measure one. \\
\indent Next, using the definition of PD and arguments similar to those used above, we get that for any ${\bf x} \in l_{2}$,
\begin{eqnarray}
\frac{1 - PD({\bf x})}{PD({\bf x})} \ = \ \sup_{{\bf u} \in l_{2}} \frac{|\langle{\bf u},{\bf x}\rangle - E(\langle{\bf u},{\bf X}\rangle)|}{\sqrt{Var(\langle{\bf u},{\bf X}\rangle)}}  &\geq& \sup_{d \geq 1} \sup_{{\bf v} \in \mathbb{R}^{d}} \frac{|{\bf v}'{\bf x}(d) - E({\bf v}'{\bf X}(d))|}{\sqrt{Var({\bf v}'{\bf X}(d))}}   \nonumber \\
\hspace{1cm} \geq \ \sup_{d \geq 1} \sup_{{\bf v} \in \mathbb{R}^{d}} \frac{|{\bf v}'{\bf y}(d)|}{\sqrt{Var({\bf v}'{\bf Y}(d))}} &\geq& \sup_{d \geq 1} \sup_{{\bf v} \in \mathbb{R}^{d}} \frac{\left|\sum_{k=1}^{d} v_{k}y_{k}\right|}{\sqrt{\sum_{k=1}^{d} v_{k}^{2}\tau_{k}^{2}}} \nonumber  \\
&=& \sup_{d \geq 1} \sum_{k=1}^{d} \frac{y_{k}^{2}}{\tau_{k}^{2}}.  \label{eq2.1.5}
\end{eqnarray}
As in the case of HD, in view of the moment and the mixing conditions on the $Y_{k}$'s assumed in the theorem, (\ref{eq2.1.4}) and (\ref{eq2.1.5}) now imply that $PD({\bf x}) = 0$ for all ${\bf x}$ in a subset of $l_{2}$ with $\mu$-measure one.    
\end{proof}
\vspace{-0.1in}
\begin{proof}{\it of Theorem \ref{thm2.2}} \ \
Let us denote the dual space of $C[0,1]$ by ${\cal M}[0,1]$. Consider the measure ${\bf u}_{d} \in {\cal M}[0,1]$, which assigns point mass $v_{k}$ at $k/d$, $k = 1,2,\ldots,d$. So, we have ${\bf u}_{d}({\bf x}) = \sum_{k=1}^{d} v_{k}x_{k/d}$ for any ${\bf x} = \{x_{t}\}_{t \in [0,1]} \in C[0,1]$. Let ${\bf v} = (v_{1},v_{2},\ldots,v_{d})'$, ${\bf X}_{d} = (X_{1/d},X_{2/d},\ldots,X_{d/d})'$ and ${\bf x}_{d} = (x_{1/d},x_{2/d},\ldots,x_{d/d})'$. For each $d \geq 1$, define $Y_{d,1} = X_{1/d} - E(X_{1/d})$, and let $Y_{d,k}$ denote the residual of linear regression of $X_{k/d}$ on $(X_{1/d},X_{2/d},\ldots,X_{(k-1)/d})$ for $k = 2,3,\ldots,d$. Then, ${\bf Y}_{d} = (Y_{d,1},Y_{d,2},\ldots,Y_{d,k})'$ has a multivariate Gaussian distribution with independent components in view of the Gaussian distribution of ${\bf X}$. The proof now follows by straightforward modification of the arguments used in the proof of Theorem \ref{thm2.1} and using ${\bf Y}_{d}$ in place of ${\bf Y}(d)$.  
\end{proof}
\vspace{-0.1in}
\begin{proof}{\it of Theorem \ref{thm3.1}} \ \
Let $\widetilde{{\bf X}} = (\widetilde{X}_{1},\widetilde{X}_{2},\ldots)$ and ${\bf X}_{i} = (X_{i,1},X_{i,2},\ldots)$, $i=1,2,\ldots,J$, be independent copies of ${\bf X}$. We first note that $BD({\bf x}) = HRD({\bf x}) = 0$ with probability one iff $E\{BD(\widetilde{{\bf X}})\} = E\{HRD(\widetilde{{\bf X}})\} = 0$. Let us first consider the case of BD. Note that $E\{BD(\widetilde{{\bf X}})\} = \sum_{j=2}^{J} P(\min_{1 \leq i \leq j} X_{i,k} \leq \widetilde{X}_{k} \leq \max_{1 \leq i \leq j} X_{i,k}, \ \forall \ k \geq 1)$. So, $E\{BD(\widetilde{{\bf X}})\} = 0$ iff $P(\min_{1 \leq i \leq j} X_{i,k} \leq \widetilde{X}_{k} \leq \max_{1 \leq i \leq j} X_{i,k}, \ \forall \ k \geq 1) = 0$ for all $2 \leq j \leq J$. Consequently, it is enough to show that for any $2 \leq j \leq J$, the event $\{\min_{1 \leq i \leq j} X_{i,k} > \widetilde{X}_{k}\} \cup \{\max_{1 \leq i \leq j} X_{i,k} < \widetilde{X}_{k}\}$ occurs for some $k \geq 1$ with probability one. Now, the sequence $(\min_{1 \leq i \leq j} X_{i,1} - \widetilde{X}_{1},\min_{1 \leq i \leq j} X_{i,2} - \widetilde{X}_{2},\ldots)$ is $\alpha$-mixing for any $1 \leq j \leq J$, and its mixing coefficients satisfy the conditions assumed in the theorem. On the other hand, $P(\min_{1 \leq i \leq j} X_{i,k} > \widetilde{X}_{k}) = 2^{-j}$ for all $k \geq 1$, by the continuity of the distributions of the $X_{k}$'s. So, using Corollary 4 in \cite{Hans91}, we have $d^{-1} \sum_{k=1}^{d} I(\min_{1 \leq i \leq j} X_{i,k} > \widetilde{X}_{k}) \rightarrow 2^{-j}$ as $d \rightarrow \infty$ with probability one for all $1 \leq j \leq J$. So, the event $\{\min_{1 \leq i \leq j} X_{i,k} > \widetilde{X}_{k}\}$ actually occurs for infinitely many $k \geq 1$ with probability one. Thus, $BD({\bf x}) = 0$ for all ${\bf x}$ in a subset of $l_{2}$ with $\mu$-measure one. \\
\indent The proof for HRD follows by taking $j = 1$, and we skip further details.  
\end{proof}
\vspace{-0.1in}
\begin{lemma} \label{lem6.4}
Let $\{X_{t}\}_{t \in [0,1]}$ be a Feller processes in $C[0,1]$ satisfying the conditions of Theorem \ref{thm3.2}. Let ${\bf X}_{i} = \{X_{i,t}\}_{t \in [0,1]}$, $i = 1,2,\ldots,J$, denote independent copies of ${\bf X}$, and define $T_{j} = \inf\{t > 0 : \min_{1 \leq i \leq j} X_{i,t} > x_{0}\}$ and $S_{j} = \inf\{t > 0 : \max_{1 \leq i \leq j} X_{i,t} < x_{0}\}$ for $1 \leq j \leq J$. Then, $P(T_{j} = 0) = P(S_{j} = 0) = 1$ for all $1 \leq j \leq J$. 
\end{lemma}
\vspace{-0.1in}
\begin{proof}
Consider the multivariate Feller process $\{(X_{1,t},X_{2,t},\ldots,X_{j,t})\}_{t \in [0,1]}$, where $1 \leq j \leq J$. Since, $P(T_{j} \leq t) \geq P(\min_{1 \leq i \leq j} X_{i,t} > x_{0}) = 2^{-j}$ and $P(S_{j} \leq t) \geq P(\max_{1 \leq i \leq j} X_{i,t} < x_{0}) = 2^{-j}$ for every $t > 0$, we have
\begin{eqnarray}
P(T_{j} = 0) = \lim_{s \downarrow 0} P(T_{j} \leq s) \geq 2^{-j} \ \ \mbox{and} \ \ P(S_{j} = 0) = \lim_{s \downarrow 0} P(S_{j} \leq s) \geq 2^{-j}.   \label{eq6.4.1}
\end{eqnarray}   
From the continuity of the sample paths of the processes, and using Propositions 2.16 and 2.17 in \cite{RY91}, it follows that $P(T_{j} = 0) = 0$ or $1$ and $P(S_{j} = 0) = 0$ or $1$ for all $1 \leq j \leq J$. The proof is now complete using (\ref{eq6.4.1}).  
\end{proof}
\vspace{-0.1in}
\begin{lemma} \label{lem6.5}
Let $\{X_{t}\}_{t \in [0,1]}$ be a Feller process on $C[0,1]$ satisfying the conditions of Theorem \ref{thm3.2}. Also, let ${\bf f} = \{f_{t}\}_{t \in [0,1]} \in C[0,1]$ be such that $f_{0} = x_{0}$ and $f_{t} - x_{0}$ changes sign infinitely often in any right neighbourhood of zero. Then, $P(T = 0) = P(S = 0) = 1$, where $T = \inf\{t > 0 : X_{t} - f_{t} > 0\}$ and $S = \inf\{t > 0 : X_{t} - f_{t} < 0\}$.
\end{lemma}
\begin{proof}
For any $t > 0$, let $0 < r < t$ be such that $f_{r} < x_{0}$. Then, $P(T \leq t) \geq P(T \leq r) \geq P(X_{r} > f_{r}) \geq P(X_{r} > x_{0}) = 1/2$. Now, arguing as in the proof of Lemma \ref{lem6.4}, we get that $P(T = 0) = 1$ since $\{X_{t}-f_{t}\}_{t \in [0,1]}$ is a Feller proces staring at $0$. Next, let $0 < s < t$ be such that $f_{s} > x_{0}$. By similar arguments, we get that $P(S = 0) = 1$.   
\end{proof}
\vspace{-0.1in}
\begin{proof}{\it of Theorem \ref{thm3.2}} \ \ 
We first prove the result for BD using similar ideas as in the proof of Theorem \ref{thm3.1}. From the definition of BD in (\ref{s3.1}), we have 
\begin{eqnarray}
E\{BD(\widetilde{{\bf X}})\} &=& \sum_{j=2}^{J} P\left(\min_{1 \leq i \leq j} X_{i,t} \leq \widetilde{X}_{t} \leq \max_{1 \leq i \leq j} X_{i,t}, \ \forall \ t \in [0,1]\right)   \nonumber \\
&\leq& \sum_{j=2}^{J} P\left(\min_{1 \leq i \leq j} X_{i,t} \leq \widetilde{X}_{t}, \ \forall \ t \in [0,1]\right)  \nonumber \\
&=& \sum_{j=2}^{J} \left. E\left\{P\left(\min_{1 \leq i \leq j} X_{i,t} \leq \widetilde{X}_{t}, \ \forall \ t \in [0,1] ~ \right|~ {\bf X}_{1}, {\bf X}_{2}, \ldots, {\bf X}_{J} \right) \right\}.    \label{eq3.2.1}
\end{eqnarray}
For any fixed $j$, let ${\bf z} = \{z_{t}\}_{t \in [0,1]}$ be a realization of the process $\{\min_{1 \leq i \leq j} X_{i,t}\}_{t \in [0,1]}$. Then, from Lemma \ref{lem6.4}, it follows that ${\bf z}$ satisfies, with probability one, the assumptions made on the function ${\bf f}$ in Lemma \ref{lem6.5}. So, using Lemma \ref{lem6.5}, we have $P(z_{t} \leq \widetilde{X}_{t}, \ \forall \ t \in [0,1]) = 0$ for all ${\bf z}$ in a set of probability one. Hence, the expectation in (\ref{eq3.2.1}) is zero, which implies that $E\{BD(\widetilde{{\bf X}})\} = 0$. Thus, $BD({\bf x}) = 0$ on a set of $\mu$-measure one. \\
\indent The proof for HRD follows by taking ${\bf z}$ to be a realization of the process ${\bf X}$, and using Lemma \ref{lem6.4} and similar arguments as above.  
\end{proof}
\vspace{-0.1in}
\begin{lemma} \label{lem6.1}
Let ${\bf G}$ be the map on $C[0,1]$ defined as ${\bf G}({\bf f}) = \{g(t,f_{t})\}_{t \in [0,1]}$, where ${\bf f} = \{f_{t}\}_{t \in [0,1]} \in C[0,1]$ and $g : [0,1] \times \mathbb{R} \rightarrow \mathbb{R}$ is continuous. Then, ${\bf G}$ is a continuous map from $C[0,1]$ into $C[0,1]$.
\end{lemma}
\begin{proof}
Let $t_{n} \rightarrow t$ in $[0,1]$ as $n \rightarrow \infty$. By the continuity of $g$, and the fact that ${\bf f} = \{f_{t}\}_{t \in [0,1]} \in C[0,1]$, we have $g(t_{n},f_{t_{n}}) \rightarrow g(t,f_{t})$ as $n \rightarrow \infty$. This shows that ${\bf G}$ maps $C[0,1]$ into $C[0,1]$. Let us now fix $\epsilon > 0$, $t \in [0,1]$ and ${\bf f} \in C[0,1]$. Consider a sequence of functions ${\bf f}_{n} = \{f_{n,t}\}_{t \in [0,1]}$ in $C[0,1]$ such that $||{\bf f}_{n} - {\bf f}|| \rightarrow 0$ as $n \rightarrow \infty$. Note that the function $g$ is uniformly continuous on $[0,1] \times I$, where $I$ is any compact interval of the real line. Thus, $\sup_{t \in [0,1]}|g(t,f_{n,t}) - g(t,f_{t})| \rightarrow 0$, and this proves the continuity of ${\bf G}$. 
\end{proof}
\vspace{-0.1in}
\begin{proof}{\it of Theorem \ref{thm3.3}} \ \ 
(a) Since the process ${\bf Y} = \{Y_{t}\}_{t \in [0,1]}$ has almost surely continuous sample paths, Lemma \ref{lem6.1} implies that the sample paths of the process ${\bf X} = {\bf G}({\bf Y})$ also lie in $C[0,1]$ almost surely. Consider now ${\bf x}_{p} = {\bf G}({\bf y}_{p})$, where $p \in (0,1)$ and ${\bf y}_{p} = \{F_{t}^{-1}(p)\}_{t \in [0,1]}$. Note that the distribution $F_{t}$ of $Y_{t}$ is Gaussian for all $t \in (0,1]$ with zero mean and variance $\sigma_{t}^{2}$ (say), which is a continuous function in $t$. So, $F_{t}^{-1}(p) = \sigma_{t}\Phi^{-1}(\zeta_{p})$, where $\Phi$ and $\zeta_{p}$ denote the distribution function and the $p$th quantile of the standard normal variable, respectively. Hence, ${\bf y}_{p} \in C[0,1]$, and in view of Lemma \ref{lem6.1}, we have ${\bf x}_{p} = {\bf G}({\bf y}_{p}) \in C[0,1]$.  \\
\indent Note that by strict monotonicity of $g(t,.)$ for all $t \in [0,1]$, we have $MBD({\bf x}_{p}) = \sum_{j=2}^{J} [1 - p^{j} - (1-p)^{j}]$, $MHRD({\bf x}_{p}) = \min(p,1-p)$ and $ID({\bf x}_{p}) = \psi(p)$. These depth functions are bounded above by $A_{J} = J - 2 + 2^{-J+1}$, $1/2$ and $\sup_{s \in (0,1)} \psi(s)$, respectively, where the upper bounds are attained in MBD and MHRD iff $p = 1/2$.  Let us now write $C_{y_{0}}[0,1] = \{{\bf f} = \{f_{t}\}_{t \in [0,1]} \in C[0,1] : f_{0} = y_{0}\}$, and define $H_{0} = {\bf G}(C_{y_{0}}[0,1]) = \{{\bf G}({\bf f}) : {\bf f} \in C_{y_{0}}[0,1]\}$. Since ${\bf x}_{p} \in H_{0}$, we have $MBD(H_{0}) = \{MBD({\bf x}) : {\bf x} \in H_{0}\} = (0,A_{J}]$, $MHRD(H_{0}) = (0,1/2]$ and $ID(H_{0}) = \psi((0,1))$ by varying $p \in (0,1)$. This completes the proof of part (a). \\
\indent (b) It follows from the proof of Proposition 5.1 in \cite{Guas06} that the support of a fractional Brownian motion, say $\{Z_{t}\}_{t \in [0,1]}$, starting at zero is the whole of $C_{0}[0,1]$. Since the distribution of $\{Y_{t}\}_{t \in [0,1]}$ is same as that of $\{Z_{t} + y_{0}\}_{t \in [0,1]}$, the support of the distribution of $\{Y_{t}\}_{t \in [0,1]}$ is the whole of $C_{y_{0}}[0,1]$. By continuity of ${\bf G}$ proved in Lemma \ref{lem6.1}, any point in $H_{0}$ is a support point of the distribution of $\widetilde{{\bf X}}$. On the other hand, for every fixed $t \in [0,1]$, since $g(t,.)$ is a continuous strictly monotone function, and the distribution of $Y_{t}$ is continuous, it follows that the distribution of $X_{t}$ is continuous. So, using the dominated convergence theorem, we get that  MBD, MHRD and ID are continuous functions on $C[0,1]$. This and the fact that any point in $H_{0}$ is a support point of the distribution of $\widetilde{{\bf X}}$ completes the proof of part (b). \\
\indent (c) If $\{Y_{t}\}_{t \in [0,1]}$ is a fractional Brownian bridge ``tied'' down to $b_{0}$ at $t = 1$ (say), then it has the same distribution as that of $\{Z_{t} - Cov(Z_{t},Z_{1})(Z_{1} - b_{0})\}_{t \in [0,1]}$. So, the support of $\{Y_{t}\}_{t \in [0,1]}$ is the set $\{{\bf f} = \{f_{t}\}_{t \in [0,1]} \in C_{y_{0}}[0,1] : f_{1} = b_{0}\}$. The proof now follows from arguments similar to those in parts (a) and (b).   
\end{proof}
\vspace{-0.1in}
\begin{remark}  \label{rem6.1}
It follows from the proof of Proposition $5.1$ in \cite{Guas06} that a fractional Brownian motion $\{Y_{t}\}_{t \in [0,1]}$ starting at $y_{0}$ has as its support as the whole of $C_{y_{0}}[0,1]$, which implies that the support of $\{Y_{t}\}_{t \in [t_{0},1]}$ is the whole of $C[t_{0},1]$ for any $t_{0} > 0$. Consequently, if MBD, MHRD and ID are computed based on the distribution of $\{X_{t}\}_{t \in [t_{0},1]}$, the supports of the distributions of $MBD(\widetilde{{\bf X}})$, $MHRD(\widetilde{{\bf X}})$ and $ID(\widetilde{{\bf X}})$ will be $[0,A_{J}]$, $[0,1/2]$ and the closure of $\psi((0,1))$, respectively.
\end{remark}
\vspace{-0.1in}
\begin{proof}{\it of Theorem \ref{thm4.1}} \ \ 
First, we shall prove that the support of $\widetilde{{\bf X}}$ is the whole of $l_{2}$, where $\widetilde{{\bf X}} = (\widetilde{X}_{1},\widetilde{X}_{2},\ldots)$ is an independent copy of ${\bf X} = (X_{1},X_{2},\ldots)$. For this, let us fix ${\bf x} \in l_{2}$ and $\eta > 0$. Then, there exists $d \geq 1$ satisfying $||{\bf x} - {\bf x}[d]|| < \eta$, where ${\bf x}[d] = (x_{1},x_{2},\ldots,x_{d},0,0,\ldots)$. Further, in view of the assumption on the second moments of the $X_{k}$'s, we can choose $M > d$ such that $\sum_{k > M} E(\widetilde{X}_{k}^{2}) < \eta^{2}/4$. Then,
\begin{eqnarray}
&& P(||\widetilde{{\bf X}} - {\bf x}|| < 2\eta) > P(||\widetilde{{\bf X}} - {\bf x}[d]|| < \eta)   \nonumber \\
&>& P\left(\sum_{k \leq M} (\widetilde{X}_{k} - x_{k})^{2} < \frac{\eta^{2}}{2} ~ \left |~ \sum_{k > M} \widetilde{X}_{k}^{2} < \frac{\eta^{2}}{2}\right) \right. \ P\left(\sum_{k > M} \widetilde{X}_{k}^{2} < \frac{\eta^{2}}{2}\right). \label{eq4.1.1}
\end{eqnarray}
Using Markov inequality, we get 
\begin{eqnarray}
 P\left(\sum_{k > M} \widetilde{X}_{k}^{2} < \frac{\eta^{2}}{2}\right) \ > \ 1 - \frac{\sum_{k > M} E(\widetilde{X}_{k}^{2})}{\eta^{2}/2} \ > \ 1/2.   \label{eq4.1.2}
\end{eqnarray}
(\ref{eq4.1.1}) and (\ref{eq4.1.2}) now imply that
\begin{eqnarray}
P(||\widetilde{{\bf X}} - {\bf x}|| < 2\eta) &>& \frac{1}{2} P\left(\sum_{k \leq M} (\widetilde{X}_{k} - x_{k})^{2} < \frac{\eta^{2}}{2} ~ \left |~ \sum_{k > M} \widetilde{X}_{k}^{2} < \frac{\eta^{2}}{2}\right) \right.. \label{eq4.1.3} 
\end{eqnarray}
From the conditional full support assumption on the $X_{k}$'s, it follows that the expression on the right hand side of the inequality (\ref{eq4.1.3}) is positive for each $\eta > 0$. This implies that ${\bf x}$ lies in the support of $\widetilde{{\bf X}}$.  \\
\indent Since the distribution of ${\bf X}$ is non-atomic, SD is a continuous function on $l_{2}$ as mentioned in Section \ref{sec:3}. Thus, the set $\{SD({\bf x}) : {\bf x} \in l_{2}\}$ is an interval in $[0,1]$. Hence, from the properties of SD discussed in Section \ref{sec:3}, we get that the function SD takes all values in $(0,1]$. This and the continuity of SD together imply that the support of the distribution of $SD(\widetilde{{\bf X}})$ is the whole of $[0,1]$.   
\end{proof}
\vspace{-0.1in}
\begin{lemma} \label{lem6.2}
The set $H_{0} = {\bf G}(C_{y_{0}}[0,1])$ is convex. Here, ${\bf G}$ is as in Lemma \ref{lem6.4} and $C_{y_{0}}[0,1]$ is as in the proof of Theorem \ref{thm3.3}.
\end{lemma}
\begin{proof}
Let us take ${\bf f} = \{f_{t}\}_{t \in [0,1]}$ and ${\bf h} = \{h_{t}\}_{t \in [0,1]} \in C_{y_{0}}[0,1]$. Fix $\lambda \in (0,1)$ and $t \in [0,1]$. Let $L = \max(||{\bf f}||, ||{\bf h}||)$. By continuity of $g(t,.)$, the range of $g(t,s)$ for $s \in [-L,L]$ is a closed and bounded interval, say $[a,b]$. Thus, $\lambda g(t,f_{t}) + (1-\lambda) g(t,h_{t}) \in [a,b]$. Since $g(t,.)$ is continuous and strictly increasing, there is a unique $q_{t} \in [-L,L]$ such that $g(t,q_{t}) = \lambda g(t,f_{t}) + (1-\lambda) g(t,h_{t})$. Now let $t_{n} \rightarrow t \in [0,1]$ as $n \rightarrow \infty$. Since $g(t_{n},q_{t_{n}}) = \lambda g(t_{n},f_{t_{n}}) + (1-\lambda) g(t_{n},h_{t_{n}})$, by continuity of $g$, we have
\begin{eqnarray}
g(t_{n},q_{t_{n}}) \rightarrow \lambda g(t,f_{t}) + (1-\lambda) g(t,h_{t}) = g(t,q_{t})   \label{eq6.2.1}
\end{eqnarray}
as $n \rightarrow \infty$. Suppose now, if possible, $q_{t_{n}} \nrightarrow q_{t}$ as $n \rightarrow \infty$. Then, there exists $\epsilon_{0} > 0$ and a subsequence $\{t_{n_{j}}\}_{j \geq 1}$ such that $|q_{t_{n_{j}}} - q_{t}| > \epsilon_{0}$ for all $j \geq 1$. A further subsequence of $\{t_{n_{j}}\}_{j \geq 1}$ will converge to some $b_{t} \in [-L,L]$, and hence, $|b_{t} - q_{t}| \geq \epsilon_{0}$. Along that latter subsequence, we have $g(t_{n_{j}},q_{t_{n_{j}}})$ converging to $g(t,b_{t})$. This and (\ref{eq6.2.1}) together imply that $g(t,b_{t}) = g(t,q_{t})$. So, by strict monotonicity of $g(t,.)$, we get that $b_{t} = q_{t}$, which yields a contradiction. Hence, $q_{t_{n}} \rightarrow q_{t}$ as $n \rightarrow \infty$, which implies that ${\bf q} = \{q_{t}\}_{t \in [0,1]} \in C_{y_{0}}[0,1]$. This proves the convexity of $H_{0}$. 
\end{proof}
\vspace{-0.1in}
\begin{lemma} \label{lem6.3}
Every point in $H_{0}$ is a support point of the distribution of $\widetilde{{\bf X}}$ in $L_{2}[0,1]$. Here $\widetilde{{\bf X}}$ is as in Theorem \ref{thm4.2}.
\end{lemma}
\vspace{-0.1in}
\begin{proof}
Fix ${\bf f} \in C_{y_{0}}[0,1]$ and $\eta > 0$. Let $||.||$ denote the supremum norm on $C[0,1]$ as before, and $||.||_{2}$ denote the usual norm on $L_{2}[0,1]$. Since $||{\bf y}||_{2} \leq ||{\bf y}||$ for any ${\bf y} \in C[0,1]$,  we have $P(||{\bf G}({\bf Y}) - {\bf G}({\bf f})||_{2} < \eta) > P(||{\bf G}({\bf Y}) - {\bf G}({\bf f})|| < \eta)$. By the continuity of ${\bf G}$ proved in Lemma \ref{lem6.1}, there exists $\delta > 0$ depending on $\eta$ and ${\bf f}$ such that $P(||{\bf G}({\bf Y}) - {\bf G}({\bf f})|| < \eta) > P(||{\bf Y} - {\bf f}|| < \delta)$. Since any element in $C_{y_{0}}[0,1]$ is a support point of the distribution of ${\bf Y}$ in $C[0,1]$, we have $P(||{\bf Y} - {\bf f}|| < \delta) > 0$. It now follows that ${\bf G}({\bf f}) \in H_{0}$ is a support point of the distribution of $\widetilde{{\bf X}} = {\bf G}(\widetilde{{\bf Y}})$ in $L_{2}[0,1]$, where $\widetilde{{\bf Y}}$ denotes an independent copy of ${\bf Y}$. This completes the proof.  
\end{proof} 
\vspace{-0.1in}
\begin{proof}{\it of Theorem \ref{thm4.2}} \ \
We will first show that $SD({\bf x})$ takes all values in $(0,1)$ as ${\bf x}$ varies in $C[0,1]$. As discussed in Section \ref{sec:3}, the spatial depth function is continuous on $L_{2}[0,1]$. We have $H_{0} \subseteq C[0,1] \subseteq L_{2}[0,1]$, and $H_{0}$ is convex by Lemma \ref{lem6.2}, which implies that the set $SD(H_{0}) = \{SD({\bf f}) : {\bf f} \in H_{0}\}$ is  an interval in $[0,1]$. It follows from the non-atomicity of ${\bf X}$ and Lemma 4.14 in \cite{Kemp87} that $SD({\bf m}) = 1$, where ${\bf m}$ is a spatial median of ${\bf X}$ in $L_{2}[0,1]$. Further, from Remark $4.20$ in \cite{Kemp87}, it follows that ${\bf m}$ lies in the closure of $H_{0}$ in $L_{2}[0,1]$. Thus, there exists a sequence $\{{\bf m}_{n}\}_{n \geq 1}$ in $H_{0} \subseteq C[0,1]$ such that $||{\bf m}_{n} - {\bf m}||_{2} \rightarrow 0$ as $n \rightarrow \infty$, where $||.||_{2}$ is the usual norm in $L_{2}[0,1]$ as before. Hence, by continuity of the spatial depth function,, we have $SD({\bf m}_{n}) \rightarrow 1$ as $n \rightarrow \infty$. We next consider the sequence of linear functions $\{{\bf r}_{n}\}_{n \geq 1}$, where ${\bf r}_{n} = \{g(0,y_{0}) + d_{n}t\}_{t \in [0,1]}$ and $d_{n} \rightarrow \infty$ as $n \rightarrow \infty$. Since $g(t.,)$ is a strictly increasing continuous function for each $t \in [0,1]$, there exists $f_{n,t}$ such that $g(t,f_{n,t}) = g(0,y_{0}) + d_{n}t$. Using the assumptions about $g$, it can be shown that for each $n \geq 1$, the function ${\bf f}_{n} = \{f_{n,t}\}_{t \in [0,1]} \in C_{y_{0}}[0,1]$, which implies that ${\bf r}_{n} = {\bf G}({\bf f}_{n}) \in H_{0}$. Now, using dominated convergence theorem, we have $SD({\bf r}_{n}) \rightarrow 0$ as $n \rightarrow \infty$ in view of the fact that $d_{n} \rightarrow \infty$, and ${\bf r}_{n}/d_{n}$ converges to the identity function $\{t\}_{t \in [0,1]} \in C[0,1]$ as $n \rightarrow \infty$. Hence, $SD(H_{0}) \supseteq (0,1)$. Note that we will have $SD(H_{0}) = (0,1]$ if the spatial median ${\bf m}$ actually lies in $H_{0}$. Using Lemma \ref{lem6.3}, and the continuity of SD along with the fact that $SD(H_{0}) \supseteq (0,1)$, we get that the support of the distribution of $SD(\widetilde{{\bf X}})$ is the whole of $[0,1]$.    
\end{proof}

\end{document}